\documentclass{llncs}

\usepackage[utf8]{inputenc}
\usepackage{amssymb}
\usepackage{amsmath}
\usepackage{enumerate}
\usepackage{cite}
\usepackage{times}
\usepackage[T1]{fontenc}
\usepackage[utf8]{inputenc}
\usepackage{amsthm}
\usepackage{amstext}
\usepackage{authblk}
\usepackage{bbm} 
\usepackage{mathrsfs} 
\usepackage{color}

\usepackage{hyperref}

\newcommand\vp{\ensuremath{\mathsf{VP}}}
\newcommand\vnp{\ensuremath{\mathsf{VNP}}}
\newcommand\p{\ensuremath{\mathsf{P}}}
\newcommand\np{\ensuremath{\mathsf{NP}}}

\newcommand{\N}{\mathbb{N}}

\newcommand{\Z}{\mathbb{Z}}

\newcommand{\R}{\mathbb{R}}

\renewcommand{\leq}{\leqslant}

\renewcommand{\geq}{\geqslant}

\title{Log-concavity and lower bounds for arithmetic circuits}
\author{Ignacio Garc\'{\i}a-Marco\inst{1}\thanks{This work was
    supported by ANR  project CompA (project number: ANR-13-BS02-0001-01).} \and Pascal
  Koiran\inst{1}$^\star$ \and S\'ebastien Tavenas\inst{2}$^\star$}
\institute{LIP, ENS Lyon, France \thanks{UMR 5668 ENS Lyon - CNRS - UCBL - INRIA, Universit\'e de Lyon}\\
  \email{ignacio.garcia-marco@ens-lyon.fr, pascal.koiran@ens-lyon.fr}
  \and Max-Planck-Insitut f\"ur
  Informatik, Saarbr\"{u}cken, Germany \\ \email{ stavenas@mpi-inf.mpg.de}}

\begin{document}

\maketitle

\begin{abstract}
One question that we investigate in this paper is, how can we build
log-concave polynomials using sparse polynomials as building blocks?
More precisely, let $f = \sum_{i = 0}^d a_i X^i \in \R^+[X]$ be a
polynomial satisfying the log-concavity condition $a_i^2 > \tau
a_{i-1}a_{i+1}$ for every $i \in \{1,\ldots,d-1\},$ where $\tau >
0$. Whenever $f$ can be written under the form $f = \sum_{i = 1}^k
\prod_{j = 1}^m f_{i,j}$ where the polynomials $f_{i,j}$ have at
most $t$ monomials, it is clear that $d \leq k t^m$. Assuming that
the $f_{i,j}$ have only non-negative coefficients, we improve this
degree bound to $d = \mathcal O(k m^{2/3} t^{2m/3} {\rm log^{2/3}}(kt))$ if
$\tau > 1$, and to $d \leq kmt$ if $\tau = d^{2d}$.

This investigation has a complexity-theoretic motivation:
we show that a suitable strengthening of the above results would imply
a separation of the algebraic complexity classes $\vp$ and $\vnp$.
 As they currently stand, these
results are strong enough to provide a new example of a family of
polynomials in $\vnp$ which cannot be computed by monotone arithmetic
circuits of polynomial size.
\end{abstract}

\section{Introduction}\label{introduction}


Let $f = \sum_{j = 0}^d a_j X^j\in \R[X]$ be a univariate polynomial
of degree $d \in \mathbb{Z}^+$.  It is a classical result due to
Newton (see \cite{HLP}, \S2.22 and \S4.3 for two proofs) that
whenever all the roots of $f$ are real, then the coefficients of $f$
satisfy the following log-concavity condition:

\begin{equation} a_i^2 \geq \frac{d-i+1}{d-i} \frac{i+1}{i}\, a_{i-1} a_{i+1}
{\rm \ for \ all \ } i \in
\{1,\ldots,d-1\}.\label{newton}\end{equation} Moreover, if the roots
of $f$ are not all equal, these inequalities are strict. When $d =
2$, condition (\ref{newton}) becomes $a_1 \geq 4 a_0 a_2$, which is
well known to be a necessary and sufficient condition for all the
roots of $f$ to be real. Nevertheless, for $d \geq 3$, the converse
of Newton's result does not hold any more~\cite{Kurtz}.

\medskip

When $f \in \R^+[X]$, i.e., when $f = \sum_{j = 0}^d a_j X^j$ with
$a_j \geq 0$ for all $j \in \{0,\ldots,d\}$, a weak converse of
Newton's result holds true. Namely, a sufficient condition for $f$
to only have real (and distinct) roots is that
$$a_i^2
> 4 a_{i-1} a_{i+1} {\rm \ for\  all}\ i \in \{1,\ldots,d-1\}.$$
Whenever a polynomial fulfills this condition, we
say that it satisfies the {\it Kurtz condition} since this converse
result is often attributed to Kurtz~\cite{Kurtz}.
Note however that it was obtained some 70 years earlier by Hutchinson~\cite{Hutchinson}.
\medskip

If $f$ satisfies the Kurtz condition, all of its $d+1$ coefficients
are nonzero except possibly the constant term. Such a polynomial is
therefore very far from being sparse (recall that a polynomial is
informally called {\em sparse} if the number of its nonzero coefficients
is small compared to its degree).
One question that we investigate in this paper is: how can we
construct polynomials satisfying the Kurtz condition using sparse
polynomials as building blocks?
More precisely, consider $f$ a polynomial of the form
\begin{equation}\label{sumprod} f = \sum_{i = 1}^k \prod_{j = 1}^m
f_{i,j}\end{equation} where $f_{i,j}$ are polynomials with at most
$t$ monomials each. By expanding the products in~(\ref{sumprod}) we
see that $f$ has at most $k t^m$ monomials.
As a result, $d \leq k t^m$ if $f$ satisfies the Kurtz condition.
Our goal is to improve this very coarse bound. For the case of
polynomials $f_{i,j}$ with nonnegative coefficients, we obtain the
following result.

\begin{theorem}\label{bound}
Consider a polynomial $f \in \R^+[X]$ of degree $d$ of the form $$f
= \sum_{i = 1}^k \prod_{j = 1}^m f_{i,j},$$ where $m \geq 2$ and the
$f_{i,j} \in \R^+[X]$ have at most $t$ monomials. If $f$ satisfies
the Kurtz condition, then $d = \mathcal O(k m^{2/3} t^{2m/3} {\rm
log^{2/3}}(kt))$.
\end{theorem}
We prove this result in Section \ref{kurtzsection}. After that, in
Section \ref{strongsection}, we study the following stronger
log-concavity condition
\begin{equation} \label{stronglogconcave}a_i^2
> d^{2d} a_{i-1} a_{i+1} {\rm \ for\  all}\ i \in \{1,\ldots,d-1\}.\end{equation} In
this setting we prove the following improved analogue of Theorem
\ref{bound}.

\begin{theorem}\label{bound2} Consider a
polynomial $f \in \R^+[X]$ of degree $d$ of the form $$f = \sum_{i =
1}^k \prod_{j = 1}^m f_{i,j},$$ where $m \geq 2$ and the $f_{i,j}
\in \R^+[X]$ have at most $t$ monomials. If $f$
satisfies~$(\ref{stronglogconcave})$, then $d \leq k m t$.
\end{theorem}

 This
investigation has a complexity-theoretic motivation: we show in
Section~\ref{complexity} that a suitable 
extension of 
Theorem~\ref{bound2} (allowing negative coefficients for the
polynomials $f_{ij}$) would imply
a separation of the algebraic complexity classes $\vp$ and $\vnp$.
The classes $\vp$ of ``easily computable polynomial families'' and
$\vnp$ of ``easily definable polynomial families'' were proposed by
Valiant~\cite{Val79} as algebraic analogues of $\p$ and $\np$. As
shown in Theorem~\ref{monotone}, Theorem~\ref{bound2} as it now
stands is strong enough to provide a new example of a family of
polynomials in $\vnp$ which cannot be computed by monotone
arithmetic circuits of polynomial size.

\section{The Kurtz log-concavity condition}\label{kurtzsection}

Our main tool in this section is a result of convex geometry
\cite{EPRS}. To state this result, we need to introduce some
definitions and notations. For a pair of planar finite sets $R, S
\subset \R^2$, the {\it Minkowski sum} of $R$ and $S$ is the set $R
+ S := \{y + z \, \vert \, y \in R, z \in S\} \subset \R^2$. A
finite set $C \subset \R^2$ is {\it convexly independent} if and
only if its elements are vertices of a convex polygon. The following
result provides an upper bound for the number of elements of a
convexly independent set contained in the Minkowski sum of two other
sets.

\begin{theorem}\cite[Theorem 1]{EPRS}\label{convex} Let $R$ and $S$ be two planar point sets with $\vert
R \vert = r$ and $\vert S \vert = s$. Let $C$ be a subset of the
Minkowski sum $R + S$. If $C$ is convexly independent we have that
$\vert C \vert = \mathcal O(r^{2/3} s^{2/3} + r + s)$.
\end{theorem}

\medskip From this result the following corollary follows easily.

\begin{corollary}\label{maxconvex}Let $R_1,\ldots,R_k,S_1,\ldots,S_k,Q_1,Q_2$ be planar point sets with $\vert
R_i \vert = r, \ \vert S_i \vert = s$ for all $i \in \{1,\ldots,k\}$,
$\lvert Q_1\rvert = q_1$
and $\lvert Q_2 \rvert = q_2$. Let $C$ be a subset of $\cup_{i = 1}^k (R_i
+ S_i) + Q_1+Q_2$. If $C$ is convexly independent, then $\vert C \vert =
\mathcal O(k r^{2/3} s^{2/3} q_1^{2/3}q_2^{2/3} + k r q_1 + k s q_2)$.
\end{corollary}
\begin{proof}We observe that $\cup_{i = 1}^k (R_i + S_i) + Q_1+Q_2 =
\cup_{i = 1}^k ((R_i+Q_1) + (S_i + Q_2))$. Therefore, we partition $C$ into
$k$ convexly independent disjoint sets $C_1,\ldots,C_k$ such that
$C_i \subset (R_i+Q_1) + (S_i + Q_2)$ for all $i \in \{1,\ldots,k\}$. Since
$\vert R_i +Q_1\vert = rq_1$ and $\vert S_i + Q_2 \vert \leq sq_2$, by Theorem
\ref{convex}, we get that $\vert C_i \vert = \mathcal O(r^{2/3}
s^{2/3} q_1^{2/3}q_2^{2/3}+ rq_1 + sq_2)$ and the result follows.
\end{proof}

\medskip

\begin{theorem}\label{bound2summands}Consider a polynomial $f \in \R^+[X]$ of degree $d$ of the form $$f =
\sum_{i = 1}^k g_i h_i,$$ where $g_i,h_i \in \R^+[X]$, the $g_i$
have at most $r$ monomials and the $h_i$ have at most $s$ monomials.
If $f$ satisfies the Kurtz condition, then $d = \mathcal O(k
r^{2/3}s^{2/3}\,{\rm log}^{2/3}(k r)+ k(r+s)\log^{1/2}(kr))$.
\end{theorem}
\begin{proof}We write $f = \sum_{i = 0}^d c_i X^i$, where $c_i > 0$ for all $i \in \{1,\ldots,d\}$ and $c_0 \geq 0$.
Since $f$ satisfies the Kurtz condition, setting $\epsilon := {\rm
log}(4)/2$ we get that
\begin{equation}\label{ineq} 2 {\rm log}(c_i) > {\rm log}(c_{i-1}) + {\rm log}(c_{i+1})
+ 2 \epsilon. \end{equation} for every $i \geq 2$. For every
$\delta_1,\ldots,\delta_{d} \in \R$, we set
$C_{(\delta_1,\ldots,\delta_d)} := \{(i,{\rm log}(c_i) + \delta_i)
\, \vert \, 1 \leq i \leq d\}$. We observe that (\ref{ineq}) implies
that $C_{(\delta_1,\ldots,\delta_d)}$ is convexly independent
whenever $0 \leq \delta_i < \epsilon$ for all $i \in
\{1,\ldots,d\}$.

\medskip
We write $g_i = \sum_{j = 1}^{r_i} a_{i,j} X^{\alpha_{i,j}}$ and
$h_i = \sum_{j = 1}^{s_i} b_{i,j} X^{\beta_{i,j}}$, with $r_i \leq
r$, $s_i \leq s$ and $a_{i,j}, b_{i,j}
> 0$ for all $i,j$. Then, $c_l = \sum_{i = 1}^k (\sum_{\alpha_{i,j_1} + \beta_{i,j_2} =
l} a_{i,j_1} b_{i,j_2})$. So, setting $M_l := {\rm max} \{a_{i,j_1}
b_{i,j_2} \, \vert \, i \in \{1,\ldots,k\}, \alpha_{i,j_1} +
\beta_{i,j_2} = l\}$ for all $l \in \{1,\ldots,d\}$, we have that
$M_l \leq c_l \leq k r M_l$, so ${\rm log}(M_l) \leq {\rm log}(c_l)
\leq {\rm log}(M_l) + {\rm log}(k r)$.

\medskip For every $l \in \{1,\ldots,d\}$, we set
\begin{equation}\label{lambdaepsilon} \lambda_l := \left\lceil \frac{{\rm log}(c_l) - {\rm
log}(M_l)}{\epsilon} \right\rceil {\rm \ and \ } \delta_l := {\rm
log}(M_l) + \lambda_l \epsilon - {\rm log}(c_l),\end{equation} and
have that $0 \leq \lambda_l \leq \lceil ({\rm log}(k r))/\epsilon
\rceil$ and that $0 \leq \delta_l < \epsilon$.

\medskip
Now, we consider the sets
\begin{itemize}
\item $R_i :=
  \{(\alpha_{i,j}, {\rm log}(a_{i,j}))\, \vert \, 1 \leq j \leq r_i\}$
  for $i = 1,\ldots,k$,
\item $S_i := \{(\beta_{i,j}, {\rm
    log}(b_{i,j}))\, \vert \, 1 \leq j \leq s_i\}$ for $i = 1,\ldots,k$,
\item $Q := \{(0, \lambda \epsilon) \, \vert \, 0 \leq \lambda
  \leq \lceil {\rm log}(k r) / \epsilon \rceil \}$,
\item $Q_1 := \{(0, \mu \epsilon) \, \vert \, 0 \leq \mu
  \leq \lceil \sqrt{\log(k r) / \epsilon} \rceil \}$, and
\item $Q_2 := \{(0, \nu \lceil \sqrt{\log(k r) / \epsilon} \rceil \epsilon) \, \vert \, 0 \leq \nu
  \leq \lceil \sqrt{\log(k r) / \epsilon} \rceil \}$.
\end{itemize}
If $(0,\lambda\epsilon)\in Q$, then there exist $\mu$ and $\nu$ such that
$\lambda=\nu\lceil\sqrt{\log(kr) / \epsilon}\rceil + \mu$ where
$\mu,\nu\leq\lceil\sqrt{\log(kr) / \epsilon}\rceil$. We have,
\begin{align*}
  (0,\lambda\epsilon)=
  (0,\nu\lceil\sqrt{\log(kr) / \epsilon}\rceil\epsilon)+(0,\mu\epsilon)\in Q_1+Q_2,
\end{align*}
so $Q\subset Q_1+Q_2$.
Then, we claim that $C_{(\delta_1,\ldots,\delta_d)} \subset \cup_{i = 1}^k
(R_i + S_i) + Q$. Indeed, for all $l \in \{1,\ldots,d\}$, by
(\ref{lambdaepsilon}),  $${\rm log}(c_l) + \delta_l = {\rm log}(M_l)
+ \lambda_l \epsilon = {\rm log}(a_{i,j_1}) + {\rm log}(b_{i,j_2}) +
\lambda_l \epsilon$$ for some $i \in \{1,\ldots,k\}$ and some
$j_1,j_2$ such that $\alpha_{i,j_1} + \beta_{i,j_2} = l$; thus
$$(l,{\rm log}(c_l) + \delta_l) = (\alpha_{i,j_1}, {\rm log}(a_{i,j_1})) +
(\beta_{i,j_2}, {\rm log}(b_{i,j_1})) + (0, \lambda_l \epsilon) \in
\cup_{i = 1}^k (R_i + S_i) + Q.$$ Since
$C_{(\delta_1,\ldots,\delta_d)}$ is a convexly independent set of
$d$ elements contained in $\cup_{i = 1}^k (R_i + S_i) + Q_1+Q_2$, a direct
application of Corollary \ref{maxconvex} yields the result.
\end{proof}

\medskip

 From this result it is easy to derive an upper bound for the
general case, where we have the products of $m \geq 2$ polynomials.
If suffices to divide the $m$ factors into two groups of
approximately $m/2$ factors, and in each group we expand the product
by brute force.

\medskip

\begin{proof}[Proof of Theorem~\ref{bound}]
We write each of the $k$ products as a product of two
polynomials $G_i := \prod_{j = 1}^{\lfloor m/2 \rfloor} f_{i,j}$ and
$H_i := \prod_{j = \lfloor m/2 \rfloor + 1}^{m} f_{i,j}$. We can now
apply Theorem \ref{bound2summands} to $f = \sum_{i = 1}^k G_i H_i$
with $r = t^{\lfloor m/2 \rfloor}$ and $s = t^{m - \lfloor m/2
\rfloor}$ and we get the result.
\end{proof}

\bigskip

\begin{remark}We observe that the role of the constant $4$ in the
Kurtz condition can be played by any other constant $\tau > 1$ in
order to obtain the conclusion of Theorem \ref{bound}, i.e., we
obtain the same result for $f = \sum_{i = 0}^d a_i X^i$
satisfying that $a_i^2 > \tau a_{i-1} a_{i+1}$ for all $i \in
\{1,\ldots,d-1\}$. For proving this it suffices to replace the value
$\epsilon = {\rm log}(4) / 2$ by $\epsilon = {\rm log}(\tau) / 2$ in
the proof of Theorem \ref{bound2summands} to conclude this more
general result.
\end{remark}

\bigskip

For $f = g h$ with $g,h \in \R^+[X]$ with at most $t$ monomials,
whenever $f$ satisfies the Kurtz condition, then $f$ has only real
(and distinct) roots and so do $g$ and $h$. As a consequence, both
$g$ and $h$ satisfy (\ref{newton}) with strict inequalities and we
derive that $d \leq 2t$. Nevertheless, in the similar setting where
$f = g h + x^i$ for some $i > 0$, the same argument does not apply
and a direct application of Theorem \ref{bound} yields $d = \mathcal
O(t^{4/3}\, {\rm log^{2/3}}(t))$, a bound
which seems to be very far from optimal. 

\subsection*{Comparison with the setting of Newton polygons}

A result similar to Theorem~\ref{bound} was obtained in~\cite{KPTT}
for the Newton polygons
of bivariate polynomials. Recall that the Newton polygon of a
polynomial $f(X,Y)$ is the convex hull of the points $(i,j)$ such
that the monomial $X^iY^j$ appears in $f$ with a nonzero coefficient.
\begin{theorem}[Koiran-Portier-Tavenas-Thomass\'e] \label{mpolys}
Consider   a bivariate polynomial
 of the form
\begin{equation} \label{bivariateSPS}
f(X,Y)=\sum_{i=1}^k \prod_{j=1}^m f_{i,j}(X,Y)
\end{equation}
where $m \geq 2$ and the $f_{i,j}$ have at most $t$ monomials. The
Newton polygon  of  $f$ has $O(k t^{2m/3})$ edges.
\end{theorem}
In the setting of Newton polygons, the main issue is how to deal with
the cancellations arising from the addition of the $k$ products
in~(\ref{bivariateSPS}).
Two monomials of the form $cX^iY^j$ with the
same pair $(i,j)$ of exponents but oppositive values of the
coefficient $c$ will cancel, thereby deleting the point $(i,j)$
from the Newton polygon.

In the present paper we associate to the monomial $cX^i$ with $c>0$
the point $(i,\log c)$. There are no cancellations since we only
consider polynomials $f_{i,j}$ with nonnegative coefficients in
Theorems~\ref{bound} and~\ref{bound2summands}. However, the addition
of two monomials $cX^i, c'X^i$ with the same exponent will ``move''
the corresponding point along the coefficient axis. By contrast, in
the setting of Newton polygons points can be deleted but cannot
move. In the proof of Theorem~\ref{bound2summands} we deal with the
issue of ``movable points'' by an approximation argument, using the
fact that the constant $\epsilon=\log(4)/2>0$ gives us a little bit
of slack.

\section{A stronger log-concavity condition}\label{strongsection}

The objective of this section is to improve the bound provided in
Theorem \ref{bound} when $f = \sum_{i = 0}^d a_i X^i \in \R^+[x]$
satisfies a stronger log-concavity condition, namely, when
$a_i^2 > d^{2d} a_{i-1} a_{i+1}$ for all $i \in \{1,\ldots,d-1\}$.

To prove this bound, we make use of the following
well-known lemma (a reference and similar results for polytopes in
higher dimension can be found in~\cite{karavelas2012}).
For completeness, we provide a short proof.
\begin{lemma}\label{convexhull} If $R_1,\ldots,R_s$ are planar sets and $\vert R_i
\vert = r_i$ for all $i \in \{1,\ldots,s\}$, then the convex hull of
$R_1 + \cdots + R_s$ has at most $r_1 + \cdots + r_s$ vertices.
\end{lemma}
\begin{proof}We denote by $k_i$ the number of vertices of the
convex hull of $R_i$. Clearly $k_i \leq r_i$. Let us prove that the
convex hull of $R_1 + \cdots + R_s$ has at most $k_1 + \cdots + k_s$
vertices. Assume that $s = 2$. We write $R_1 =
\{a_1,\ldots,a_{r_1}\}$, then $a_i \in R_1$ is a vertex of the
convex hull of $R_1$ if and only if there exists $w \in S^1$ (the
unit Euclidean sphere) such that $w \cdot a_i
> w \cdot a_j$ for all $j \in \{1,\ldots,r_1\} \setminus \{i\}$.
Thus, $R_1$ induces a partition of $S^1$ into $k_1$ half-closed
intervals. Similarly, $R_2$ induces a partition of $S^1$ into $k_2$
half-closed intervals. Moreover, these two partitions induce a new
one on $S^1$ with at most $k_1 + k_2$ half-closed intervals; these
intervals correspond to the vertices of $R_1 + R_2$ and; thus, there
are at most $k_1 + k_2$. By induction we get the result for any
value of $s$.
\end{proof}

\begin{proposition}\label{SPS}
  Consider a polynomial $f=\sum_{i=0}^d a_i X^i \in \mathbb{R}^+[X]$ of the form
  \begin{align*}
    f=\sum_{i = 1}^k \prod_{j = 1}^m f_{i,j}
  \end{align*}
  where the $f_{i,j} \in \mathbb{R}^+[x]$. If $f$ satisfies the
  condition
  \begin{align*}
    a_i^2 > k^2 d^{2m} a_{i-1}  a_{i+1} ,
  \end{align*}
  then there exists a polynomial $f_{i,j}$ with at least $d / km$ monomials.
\end{proposition}

\begin{proof}
Every polynomial $f_{i,j} := \sum_{l = 0}^{d_{i,j}} c_{i,j,l}\,
X^l$, where $d_{i,j}$ is the degree of $f_{i,j}$, corresponds to a
planar set
\begin{align*} R_{i,j} := \{(l, {\rm log}(c_{i,j,l}))\, \vert \,
c_{i,j,l} > 0\} \subset \R^2.
\end{align*} We set,
$C_{i,l}  :=  {\rm max} \{0, \prod_{r = 1}^m c_{i,r,l_r} \, \vert \,
l_1 + \cdots + l_m = l \},$ for all $i \in \{1,\ldots,k\}$, $l \in
\{0,\ldots,d\}$, and $ C_l  :=  {\rm max}\{C_{i,l} \, \vert \, 1
\leq i \leq k\}$ for all $l \in \{0,\ldots,d\}$. Since the
polynomials $f_{i,j} \in \R^+[X]$  and
$$a_l = \sum_{i = 1}^k \left(\sum_{l_1 + \cdots + l_m = l}\
\prod_{r= 1}^m c_{i,r,l_r}\right)$$ for all $l \in \{0,\ldots,d\}$,
we derive the following two properties:
\begin{itemize}
\item $C_l \leq a_l \leq k d^m C_l$ for all $l \in \{0,\ldots,d\}$,
\item either $C_{i,l} = 0$ or $(l,
{\rm log}(C_{i,l})) \in R_{i,1} + \cdots + R_{i,m}$ for all $i \in
\{1,\ldots,k\}, \, l \in \{0,\ldots,d\}$. Since $a_l > 0$ for all $l
\in \{1,\ldots,d\}$, we have that $C_l > 0$ and $(l, {\rm log}(C_l))
\in \bigcup_{i = 1}^k \left(R_{i,1} + \cdots + R_{i,m}\right)$
\end{itemize}

We claim that the points in the set $\{(l, {\rm log}(C_l)) \, \vert
\, 1 \leq l \leq d\}$ belong to the upper convex envelope of
$\bigcup_{i = 1}^k (R_{i,1} + \cdots + R_{i,m})$. Indeed, if
$(a,\log(b)) \in \bigcup_{i = 1}^k (R_{i,1} + \cdots + R_{i,m})$,
then $a \in \{0,\ldots,d\}$ and $b \leq C_{a}$; moreover, for all $l
\in \{1,\ldots,d-1\}$, we have that $$C_l^2 \geq a_l^2 / (k^2
d^{2m}) > a_{l-1} \, a_{l+1} \geq C_{l-1} C_{l+1}.$$

Hence, there exist $i_0 \in \{1,\ldots,k\}$ and $L \subset
\{1,\ldots,d\}$ such that $\vert L\vert \geq d/k$ and $C_l =
C_{i_0,l}$ for all $l \in L$. Since the points in $\{(l,{\rm
log}(C_{l}))\, \vert \, 1 \leq l \leq d\}$ belong to the upper
convex envelope of $\bigcup_{i = 1}^k (R_{i,1} + \cdots +
 R_{i,m})$ we easily get that the set $\{(l, {\rm
 log}(C_{i_0,l}))  \, \vert \, l \in L\}$ is a subset of the vertices in the convex hull
of $R_{i_0,1} + \cdots  + R_{i_0,m}$. By Lemma \ref{convexhull}, we
get that there exists $j_0$ such that $\vert R_{i_0,j_0} \vert \geq
\lvert L \rvert / m \geq d / km$ points. Finally, we conclude that
$f_{i_0,j_0}$ involves at least $d / km$ monomials.
\end{proof}

\bigskip

\begin{proof}[Proof of Theorem \ref{bound2}]If $d \leq k$ or $d \leq m$, then $d \leq kmt$.
Otherwise, $d^{2d} > k^2 d^{2(d-1)} \geq k^2 d^{2m}$ and, thus, $f$
satisfies (\ref{stronglogconcave}). A direct application of
Proposition \ref{SPS} yields the result.
\end{proof}

\section{Applications to Complexity Theory}\label{complexity}

We first recall some standard definitions from algebraic complexity
theory (see e.g.~\cite{Burgi} or~\cite{Val79} for more
details). Fix a field $K$.
The elements of the complexity class $\vp$ are sequences $(f_n)$ of
multivariate polynomials with coefficients from $K$. By definition,
such a sequence  belongs to $\vp$ if the degree of $f_n$ is bounded by
a polynomial function of $n$ and if $f_n$ can be evaluated in a
polynomial number of arithmetic operations (additions and
multiplications) starting from variables and from constants in $K$.
This can be formalized with the familiar model of {\em arithmetic
  circuits}.
In such a circuit, input gates are labeled by a constant or a
variable and the other gates are labeled by an arithmetic operation
(addition or multiplication). In this paper we take $K = \R$ since
there is a focus on polynomials with nonnegative coefficients. An
arithmetic circuit is {\em monotone} if input gates are labeled by
nonnegative constants only.

A family of polynomials
belongs to the complexity class $\vnp$ if it can be
obtained  by summation from a family in $\vp$.
More precisely, $f_n(\overline{x})$ belongs to $\vnp$ if there exists
a family $(g_n(\overline{x},\overline{y}))$ in $\vp$ and a polynomial $p$
such that the tuple of variables
$\overline{y}$ is of length $l(n) \leq p(n)$ and
$$f_n(\overline{x})=\sum_{\overline{y} \in \{0,1\}^{l(n)}} g_n(\overline{x},\overline{y}).$$
Note that this summation over all boolean values of $\overline{y}$
may be of exponential size.
Whether the inclusion $\vp \subseteq \vnp$ is strict is a major open
problem in algebraic complexity.

Valiant's criterion~\cite{Burgi,Val79} shows that ``explicit''
polynomial families belong to $\vnp$. One version of it is as follows.
\begin{lemma}
Suppose that the function $\phi:\{0,1\}^* \rightarrow \{0,1\}$ is
computable in polynomial time. Then the family $(f_n)$ of multilinear
polynomials defined by
$$f_n=\sum_{e \in \{0,1\}^n} \phi(e)x_1^{e_1} \cdots x_n^{e_n}$$
belongs to $\vnp$.
\end{lemma}
Note that more general versions of Valiant's criterion are know. One
may allow polynomials with integer rather than
0/1 coefficients~\cite{Burgi}, but in Theorem~\ref{monotone}
below we will only have to deal with 0/1 coefficients.
Also, one may allow $f_n$ to depend on any (polynomially bounded)
number of variables rather than exactly $n$ variables and in this case, one may
allow the algorithm for computing the coefficients of $f_n$ to take as
input the index $n$ in addition to the tuple $e$ of exponents
(see~\cite{Koi04}, Theorem~2.3).

Reduction of arithmetic circuits to depth~4 is an important
ingredient in the proof of the forthcoming results. This phenomenon
was discovered by Agrawal and Vinay \cite{AV}. Here we will use it
under the form of \cite{Tavenas}, which is an improvement of
\cite{Koiran2012}. We will also need the fact that if the original
circuit is monotone, then the resulting depth~4 circuit is also
monotone (this is clear by inspection of the proof
in~\cite{Tavenas}). Recall that a depth~4 circuit is a sum of
products of sums of products of inputs; sum gates appear on layers
 2 and 4 and product gates on layers 1
and 3. All gates may have arbitrary fan-in.

\begin{lemma}\label{redprof4}Let $C$ be an arithmetic circuit of size $s > 1$
computing a $v$-variate polynomial of degree $d$. Then, there is an
equivalent depth $4$ circuit $\Gamma$ of size $2^{\, \mathcal
O\left(\sqrt{d \log (ds) \log (v)} \right)}$ with multiplication
gates at layer $3$ of fan-in $\mathcal O(\sqrt{d})$. Moreover, if
$C$ is monotone, then $\Gamma$ can also be chosen to be monotone.
\end{lemma}

We will use this result under the additional hypothesis that $d$ is
polynomially bounded by the number of variables $v$. In this
setting, since $v \leq s$, we get that the resulting depth $4$
circuit $\Gamma$ provided by Lemma \ref{redprof4} has size
$s^{\mathcal O(\sqrt{d})}$.


\medskip
Before stating the main results of this section, we construct an explicit family of log-concave polynomials.
\begin{lemma}\label{lem_concavVn}Let $n, s \in \Z^+$ and
consider $g_{n,s}(X) := \sum_{i=0}^{2^n-1} a_i X^i$, with
\begin{align*} a_i := 2^{si(2^n-i-1)} {\text \ for \ all \ } i \in \{0,\ldots,2^n
- 1\}. \end{align*}
 Then, $a_i^2
> 2^s \, a_{i-1} \, a_{i+1}$.
\end{lemma}
\begin{proof}Take $i \in \{1,\ldots,2^n - 2\}$, we have that
  \begin{align*}
    \log\left(2^s a_{i-1}a_{i+1}\right) & = s + s2^n(i-1) - s(i-1)i
    +s2^n(i+1) - s(i+1)(i+2) \\
    & = 2s2^n i - 2s i(i+1)- s  \\
    & < 2s 2^n i -2s i(i+1) \\
    & = \log(a_i^2).
  \end{align*}
\end{proof}

In the next theorem we start from the family $g_{n,s}$
of Lemma~\ref{lem_concavVn} and we set $s=n2^{n+1}$.

\begin{theorem}\label{ifvpvnp}
Let $(f_n) \in \N[X]$ be the family of polynomials $f_n(x)=g_{n,n2^{n+1}}(x)$.
\begin{itemize}
\item[(i)] $f_n$ has degree $2^n-1$ and satisfies the log-concavity condition {\rm (\ref{stronglogconcave})}.
\item[(ii)] If $\vp=\vnp$, $f_n$ can be written under form~{\rm
    (\ref{sumprod})} with $k=n^{O(\sqrt{n})}$,
 $m=O(\sqrt{n})$ and $t=n^{O(\sqrt{n})}$.
\end{itemize}
\end{theorem}
\begin{proof}
It is clear that $f_n \in
\N[X]$ has degree $2^n - 1$ and, by Lemma \ref{lem_concavVn}, $f_n$
satisfies (\ref{stronglogconcave}).

Consider now the related family of bivariate polynomials
$g_n(X,Y)=\sum_{i=0}^{2^n-1}X^i Y^{e(n,i)},$ where $e(n,i) = s i (2^n - i - 1)$. One can check in time polynomial in $n$ whether
a given monomial $X^iY^j$ occurs in $g_n$: we just need to check
that $i<2^n$ and that $j=e(n,i)$. By mimicking the proof of Theorem
1 in \cite{KPTT} and taking into account Lemma \ref{redprof4} we get
that, if $\vp = \vnp$, one can write
\begin{equation} \label{sumprod2}
g_n(X,Y)=\sum_{i=1}^k\prod_{j=1}^m g_{i,j,n}(X,Y)
\end{equation}
where the bivariate polynomials $g_{i,j,n}$ have $n^{O(\sqrt{n})}$
monomials, $k=n^{O(\sqrt{n})}$ and $m=O(\sqrt{n})$. Performing the
substitution $Y=2$ in~(\ref{sumprod2}) yields the required
expression for $f_n$.
\end{proof}

We believe that there is in fact no way to write $f_n$ under
form~(\ref{sumprod}) so that the parameters $k,m,t$ satisfy the
constraints $k=n^{O(\sqrt{n})}$,
 $m=O(\sqrt{n})$ and $t=n^{O(\sqrt{n})}$.
By part (ii) of Theorem~\ref{ifvpvnp}, a proof of this would
separate $\vp$ from $\vnp$. The proof of Theorem~\ref{monotone} below
shows that our belief is actually correct in the special case where
the polynomials $f_{i,j}$ in~(\ref{sumprod}) have nonnegative
coefficients.

\medskip
The main point of Theorem \ref{monotone} is to present an
unconditional lower bound for a polynomial family $(h_n)$ in $\vnp$
derived from $(f_n)$. Note that $(f_n)$ itself is not in $\vnp$
since its degree is too high. Recall that
\begin{equation} \label{fneq}
 f_n(X) := \sum_{i=0}^{2^n-1} 2^{2n2^ni(2^n-i-1)} X^i.
\end{equation}
To construct $h_n$ we write down in base 2 the exponents of ``2'' and
``$X$'' in~(\ref{fneq}).
More precisely, we take $h_n$ of the form:
\begin{equation} \label{hneq}
 h_n :=\sum_{\alpha \in \{0,1\}^{n} \atop \beta \in \{0,1\}^{4n}}
\lambda(n, \alpha, \beta)\, X_0^{\alpha_0} \cdots
X_{n-1}^{\alpha_{n-1}} Y_0^{\beta_0} \cdots
Y_{4n-1}^{\beta_{4n-1}},
\end{equation}
 where $\alpha =
(\alpha_0,\ldots,\alpha_{n-1}),\, \beta =
(\beta_0,\ldots,\beta_{4n-1})$ and $\lambda(n,\alpha,\beta) \in
\{0,1\}$; we set $\lambda(n,\alpha,\beta) = 1$ if and only
if $\sum_{j = 0}^{4n-1} \beta_j 2^j = 2n2^n i (2^n - i - 1) <
2^{4n},$ where $i := \sum_{k = 0}^{n-1} \alpha_{i,k} 2^k$.
By construction, we have:
\begin{equation} \label{transfor} f_n(X) = h_n(X^{2^0}, X^{2^1},\ldots,X^{2^{n-1}},2^{2^0},
2^{2^1},\ldots,2^{2^{4n-1}}). \end{equation}
This relation will be useful in the proof of the following lower bound theorem.
\begin{theorem}\label{monotone}
The family $(h_n)$ in~{\rm(\ref{hneq})} is in $\vnp$. If $(h_n)$ is
computed by depth $4$ monotone arithmetic circuits of size $s(n)$,
then $s(n) = 2^{\,\Omega(n)}$. If $(h_n)$ is computed by monotone
arithmetic circuits of size $s(n)$, then $s(n) =
2^{\,\Omega(\sqrt{n})}$. In particular, $(h_n)$ cannot be computed
by monotone arithmetic circuits of polynomial size.
\end{theorem}
\begin{proof}
Note that $h_n$ is a polynomial in $5n$ variables, of degree at most
$5n$, and its coefficients $\lambda(n, \alpha, \beta)$ can be
computed in polynomial time. Thus, by Valiant's criterion
we conclude that $(h_n)\in \vnp$.

Assume that $(h_n)$ can be computed by depth $4$ monotone arithmetic
circuits of size $s(n)$. Using (\ref{transfor}), we get that $f_n =
\sum_{i = 1}^k \prod_{j = 1}^m f_{i,j}$ where $f_{i,j} \in \R^+[X]$
 have at most $t$ monomials and $k,m,t$ are $\mathcal
O(s(n))$. Since the degree of $f_n$ is $2^n - 1$, by Theorem
\ref{bound2}, we get that $2^n - 1 \leq kmt$. We conclude that $s(n)
= 2^{\,\Omega(n)}$.

To complete the proof of the theorem, assume that $(h_n)$ can be
computed by monotone arithmetic circuits of size $s(n)$. By
Lemma~\ref{redprof4}, it follows that the polynomials $h_n$ are
computable by depth~4 monotone circuits of size $s'(n) :=
s(n)^{\mathcal O(\sqrt{n})}$. Therefore $s'(n) = 2^{\,\Omega(n)}$
and we finally get that $s(n) = 2^{\,\Omega(\sqrt{n})}$.
\end{proof}

Lower bounds for monotone arithmetic circuits have been known for a
long time (see for instance~\cite{jerrum82,valiant79negation}).
Theorem~\ref{monotone} provides yet another example of a polynomial
family which is hard for monotone arithmetic circuits, with an
apparently new proof method.

\section{Discussion}

As explained in the introduction, log-concavity plays a role
in the study of real roots of polynomials.
In~\cite{Koi10a} bounding the number of real roots of sums of products
of sparse polynomials was suggested as an approach for separating
$\vp$ from $\vnp$. Hrube\v{s}~\cite{Hrubes13} suggested to bound the
multiplicities of roots, and~\cite{KPTT} to bound the number of edges
of Newton polygons of bivariate polynomials.

Theorem~\ref{ifvpvnp} provides another
 plausible approach to $\vp \neq \vnp$: it
suffices to show that if a polynomial $f \in \R^+[X]$  under
form~(\ref{sumprod}) satisfies the Kurtz condition or the stronger
log-concavity condition (\ref{stronglogconcave}) then its degree is
bounded by a ``small'' function of the parameters $k,m,t$.  A degree
bound which is polynomial bound in $k,t$ and $2^m$ would be good
enough to separate $\vp$ from $\vnp$. Theorem~\ref{bound} improves
on the trivial $kt^m$ upper bound when $f$ satisfies the Kurtz
condition, but certainly falls short of this goal: not only is the
bound on $\deg(f)$ too coarse, but we would also need to allow
negative coefficients in the polynomials $f_{i,j}$.
Theorem~\ref{bound2} provides a polynomial bound on $k,m$ and $t$
under a stronger log-concavity condition, but still needs the extra
assumption that the coefficients in the polynomials $f_{i,j}$ are
nonnegative. The unconditional lower bound in Theorem~\ref{monotone}
provides a ``proof of concept'' of this approach for the easier
setting of monotone arithmetic circuits.


\bibliographystyle{plain}
\bibliography{bib}{}

\begin{thebibliography}{10}

\bibitem{AV}
Manindra Agrawal and V.~Vinay.
\newblock Arithmetic circuits: {A} chasm at depth four.
\newblock In {\em 49th Annual {IEEE} Symposium on Foundations of Computer
  Science, {FOCS} 2008, October 25-28, 2008, Philadelphia, PA, {USA}}, pages
  67--75, 2008.

\bibitem{Burgi}
Peter B{\"u}rgisser.
\newblock {\em Completeness and reduction in algebraic complexity theory},
  volume~7 of {\em Algorithms and Computation in Mathematics}.
\newblock Springer-Verlag, Berlin, 2000.

\bibitem{EPRS}
Friedrich Eisenbrand, J{\'a}nos Pach, Thomas Rothvo{\ss}, and Nir~B. Sopher.
\newblock Convexly independent subsets of the {M}inkowski sum of planar point
  sets.
\newblock {\em Electron. J. Combin.}, 15(1):Note 8, 4, 2008.

\bibitem{HLP}
G.~H. Hardy, J.~E. Littlewood, and G.~P{\'o}lya.
\newblock {\em Inequalities}.
\newblock Cambridge Mathematical Library. Cambridge University Press,
  Cambridge, 1988.
\newblock Reprint of the 1952 edition.

\bibitem{Hrubes13}
P.~Hrubes.
\newblock A note on the real $\tau$-conjecture and the distribution of complex
  roots.
\newblock {\em Theory of Computing}, 9(10):403--411, 2013.
\newblock
  \href{http://eccc.hpi-web.de/report/2012/121/}{eccc.hpi-web.de/report/2012/121/}.

\bibitem{Hutchinson}
J.~I. Hutchinson.
\newblock On a remarkable class of entire functions.
\newblock {\em Trans. Amer. Math. Soc.}, 25(3):325--332, 1923.

\bibitem{jerrum82}
Mark Jerrum and Marc Snir.
\newblock Some exact complexity results for straight-line computations over
  semirings.
\newblock {\em Journal of the ACM (JACM)}, 29(3):874--897, 1982.

\bibitem{karavelas2012}
Menelaos~I Karavelas and Eleni Tzanaki.
\newblock The maximum number of faces of the {Minkowski} sum of two convex
  polytopes.
\newblock In {\em Proceedings of the twenty-third annual {ACM-SIAM} Symposium
  on Discrete Algorithms}, pages 11--28, 2012.

\bibitem{Koi04}
P.~Koiran.
\newblock Valiant's model and the cost of computing integers.
\newblock {\em Computational Complexity}, 13:131--146, 2004.

\bibitem{Koi10a}
P.~Koiran.
\newblock Shallow circuits with high-powered inputs.
\newblock In {\em Proc. Second Symposium on Innovations in Computer Science
  (ICS 2011)}, 2011.
\newblock \href{http://arxiv.org/abs/1004.4960}{arxiv.org/abs/1004.4960}.

\bibitem{Koiran2012}
Pascal Koiran.
\newblock Arithmetic circuits: The chasm at depth four gets wider.
\newblock {\em Theoretical Computer Science}, 448(0):56 -- 65, 2012.

\bibitem{KPTT}
Pascal Koiran, Natacha Portier, S\'ebastien Tavenas, and St\'ephan Thomass\'e.
\newblock A $\tau$-conjecture for {Newton} polygons.
\newblock {\em Foundations of Computational Mathematics}, pages 1--13, 2014.

\bibitem{Kurtz}
David~C. Kurtz.
\newblock A sufficient condition for all the roots of a polynomial to be real.
\newblock {\em Amer. Math. Monthly}, 99(3):259--263, 1992.

\bibitem{Tavenas}
S.~Tavenas.
\newblock Improved bounds for reduction to depth 4 and depth 3.
\newblock In {\em Proc. 38th International Symposium on Mathematical
  Foundations of Computer Science (MFCS)}, 2013.

\bibitem{Val79}
L.~G. Valiant.
\newblock Completeness classes in algebra.
\newblock In {\em Proceedings of the Eleventh Annual ACM Symposium on Theory of
  Computing}, STOC '79, pages 249--261, New York, NY, USA, 1979. ACM.

\bibitem{valiant79negation}
Leslie~G Valiant.
\newblock Negation can be exponentially powerful.
\newblock In {\em Proceedings of the eleventh annual ACM symposium on Theory of
  computing}, pages 189--196. ACM, 1979.

\end{thebibliography}


\end{document}